\documentclass[11pt]{article}

\usepackage{amsthm,amsmath,amssymb,fullpage,bm}
\usepackage[noend]{algpseudocode}
\usepackage{algorithmicx}
\usepackage{algorithm}
\usepackage{xcolor}
\usepackage[colorlinks=true]{hyperref}

\newtheorem{theorem}{Theorem}[section]
\newtheorem{lemma}[theorem]{Lemma}
\newtheorem{definition}[theorem]{Definition}

\newtheorem{proposition}[theorem]{Proposition}

\newtheorem{corollary}[theorem]{Corollary}

\newcommand{\bmb}{{\bm b}}

\newcommand{\bme}{{\bm e}}
\newcommand{\bmx}{{\bm x}}
\newcommand{\bmy}{{\bm y}}

\newcommand{\bmw}{{\bm w}}

\newcommand{\bmzero}{{\bm 0}}
\newcommand{\bmone}{{\bm 1}}
\newcommand{\bbN}{\mathbb{N}}
\newcommand{\bbR}{\mathbb{R}}
\newcommand{\bbS}{\mathbb{S}}
\newcommand{\caC}{\mathcal{C}}
\newcommand{\caD}{\mathcal{D}}

\newcommand{\caR}{\mathcal{R}}
\newcommand{\set}[1]{\{#1\}}

\newcommand{\size}{\mathrm{size}}
\newcommand{\E}{\mathop{\mathbf{E}}}

\newcommand{\argmin}{\mathop{\mathrm{argmin}}}
\newcommand{\argmax}{\mathop{\mathrm{argmax}}}
\newcommand{\diag}{\mathrm{diag}}

\usepackage{mathtools}
\DeclarePairedDelimiter\abs{\lvert}{\rvert}

\title{Spectral Sparsification of Hypergraphs}
\author{Tasuku Soma\\
The University of Tokyo\\
\texttt{tasuku\_soma@mist.i.u-tokyo.ac.jp}
\and
Yuichi Yoshida\thanks{Supported by JSPS KAKENHI Grant Number JP17H04676}\\
National Institute of Informatics\\
\texttt{yyoshida@nii.ac.jp}
}

\begin{document}

\maketitle
\begin{abstract}
  For an undirected/directed hypergraph $G=(V,E)$, its Laplacian $L_G\colon\mathbb{R}^V\to \mathbb{R}^V$ is defined such that its ``quadratic form'' ${\bm x}^\top L_G({\bm x})$ captures the cut information of $G$.
  In particular, $\bmone_S^\top L_G(\bmone_S)$ coincides with the cut size of $S \subseteq V$, where $\bmone_S \in \mathbb{R}^V$ is the characteristic vector of $S$.

  A weighted subgraph $H$ of a hypergraph $G$ on a vertex set $V$ is said to be an $\epsilon$-spectral sparsifier of $G$ if $(1-\epsilon){\bm x}^\top L_H({\bm x}) \leq {\bm x}^\top L_G({\bm x}) \leq (1+\epsilon){\bm x}^\top L_H({\bm x})$ holds for every ${\bm x} \in \mathbb{R}^V$.
  In this paper, we present a polynomial-time algorithm that, given an undirected/directed hypergraph $G$ on $n$ vertices, constructs an $\epsilon$-spectral sparsifier of $G$ with $O(n^3\log n/\epsilon^2)$ hyperedges/hyperarcs.

  The proposed spectral sparsification can be used to improve the time and space complexities of algorithms for solving problems that involve the quadratic form, such as computing the eigenvalues of $L_G$, computing the effective resistance between a pair of vertices in $G$, semi-supervised learning based on $L_G$, and cut problems on $G$.
  In addition, our sparsification result implies that any submodular function $f\colon 2^V \to \mathbb{R}_+$ with $f(\emptyset)=f(V)=0$ can be concisely represented by a directed hypergraph.
  Accordingly, we show that, for any distribution, we can properly and agnostically learn submodular functions $f\colon 2^V \to [0,1]$ with $f(\emptyset)=f(V)=0$, with $O(n^4\log (n/\epsilon) /\epsilon^4)$ samples.
\end{abstract}

\thispagestyle{empty}
\setcounter{page}{0}
\newpage


\section{Introduction}

Let $G=(V,E,\bmw)$ be a (weighted) graph, where $\bmw \in \bbR_+^E$ is a nonnegative weight function on edges.
The \emph{Laplacian} $L_G\in \bbR^{V\times V}$ of $G$ is defined as
\[
  L_G(u,v) = \begin{cases}
  \sum\limits_{e\in E: v \in e}\bmw(e) & \text{if }u=v, \\
  -\bmw(e) & \text{if }e = \{u,v\} \in E,\\
  0 & \text{otherwise}.
  \end{cases}
\]
A notable property of the Laplacian $L_G$ is that its quadratic form can be written as
\[
  \bmx^\top L_G\bmx = \sum_{\set{u,v} \in E}\bmw(e)\bigl(\bmx(u)-\bmx(v)\bigr)^2.
\]
This quadratic form captures many interesting properties of $G$.
We say that an edge $e \in E$ is \emph{cut} by a vertex set $S \subseteq V$ if $e \cap S \neq \emptyset, e \cap (V\setminus S) \neq \emptyset$.
Then, for any vertex set $S \subseteq V$, the quantity $\bmone_S^\top L_G\bmone_S$ matches the \emph{cut weight} of $S$, that is, the total weight of edges cut by $S$, where $\bmone_S \in \bbR^V$ is the characteristic vector of $S$.
In addition, the quadratic form plays an important role in computing the eigenvalues of $L_G$ because they can be obtained by minimizing the Rayleigh quotient $\caR_G(\bmx) := \bmx^\top L_G \bmx/\|\bmx\|_2^2$ subject to some orthogonality constraints.

To reduce the time and space complexities of algorithms for solving problems that involve the quadratic form, it is desirable to approximate the quadratic form of a large input graph using that of a much smaller graph.
For $\epsilon \in (0,1)$, we say that a weighted subgraph $H$ of a graph $G$ on a vertex set $V$ is an \emph{$\epsilon$-spectral sparsifier} of $G$ if
\begin{align}
  (1-\epsilon)\bmx^\top L_H \bmx
  \leq
  \bmx^\top L_G \bmx
  \leq
  (1+\epsilon)\bmx^\top L_H \bmx
  \label{eq:spectral-sparsification-of-graphs}
\end{align}
holds for every $\bmx \in \bbR^V$~\cite{Spielman:2011kia}.
If~\eqref{eq:spectral-sparsification-of-graphs} holds (only) for characteristic vectors, $H$ is called an \emph{$\epsilon$-cut sparsifier} because it preserves the cut weights.
Note that an $\epsilon$-spectral sparsifier is also an $\epsilon$-cut sparsifier.

Cut sparsification and spectral sparsification of graphs have been studied intensively~\cite{AllenZhu:2015cm,Batson:2014io,Jambulapati:2018va,Lee:2015ex,Lee:2017bu,Spielman:2011ir}, and it is known that any graph with $n$ vertices can be spectrally sparsified with $\widetilde{O}(n/ \epsilon)$ edges~\cite{Jambulapati:2018va} or $O(n/\epsilon^2)$ edges~\cite{Lee:2017bu}, where $\widetilde{O}(\cdot)$ hides a polylogarithmic factor.

Spectral sparsifiers have numerous applications in theoretical computer science.
First, as a spectral sparsifier $H$ of $G$ preserves the cut weights of $G$, it has been used to design efficient approximation algorithms related to cuts and flows~\cite{Benczur:1996fs,Benczur:2015bm,Karger:2002du,Madry:2010fn}.
Second, we can speed up the computation of the eigenvalues of $L_G$ because, as mentioned previously, they can be computed by minimizing the Rayleigh quotient $\caR_G(\bmx)$, which can be well approximated by $\caR_H(\bmx)$.
Finally, spectral sparsification is a key technical tool for realizing nearly linear time algorithms for solving a Laplacian system, a linear system in the form of $L_G\bmx = \bmb$ for some vector $\bmb \in \bbR^V$~\cite{Koutis2014,Spielman:2014gq}.
A fast Laplacian system solver yields to a fast symmetric diagonally dominant (SDD) system solver~\cite{Spielman:2014gq}.

\subsection{Our Contribution}
In this work, we consider spectral sparsification of hypergraphs, which could be undirected or directed.
Recently, the notion of the Laplacian $L_G\colon \bbR^V \to \bbR^V$\footnote{Precisely speaking, the range of $L_G$ is $2^{\bbR^V}$, that is, a set in $\bbR^V$.
However, we simply regard its range as $\bbR^V$ because (i) $L_G(\bmx)$ is a single point almost everywhere and (ii) $\bmx^\top \bmy$ has the same value for any $\bmy \in L_G(\bmx)$. } for an undirected hypergraph $G=(V,E,\bmw)$ was introduced~\cite{Louis:2015tg,Yoshida:2017zz} such that its ``quadratic form'' $\bmx^\top L_G(\bmx) $ satisfies
\begin{align}
  \bmx^\top L_G(\bmx) = \sum_{e \in E}\bmw(e)\max_{u,v \in e}\bigl(\bmx(u)-\bmx(v)\bigr)^2
  \label{eq:quadratic-form-hypergraph}
\end{align}
for every $\bmx \in \bbR^V$.
We note that $L_G$ is no longer a linear transformation.
As with graphs, this quadratic form captures many properties of $G$.
We say that a hyperedge $e \in E$ is \emph{cut} by a vertex set $S \subseteq V$ if $e \cap S \neq \emptyset, e \cap (V\setminus S) \neq \emptyset$.
Then, we can observe that $\bmone_S^\top L_G(\bmone_S)$ is equal to the cut weight of $S$.
This quadratic form has been used to approximate the eigenvalues of $L_G$~\cite{Yoshida:2017zz} as well as in semi-supervised learning~\cite{Zhang:2017va}.

As with graphs, for $\epsilon \in (0,1)$, we say that a weighted subhypergraph $H$ of an undirected hypergraph $G$ on a vertex set $V$ is an \emph{$\epsilon$-spectral sparsifier} of $G$ if
\begin{align}
  (1-\epsilon)\bmx^\top L_H(\bmx)
  \leq
  \bmx^\top L_G(\bmx)
  \leq
  (1+\epsilon)\bmx^\top L_H(\bmx)
  \label{eq:spectral-sparsification-of-hypergraphs}
\end{align}
holds for every $\bmx \in \bbR^V$.
If~\eqref{eq:spectral-sparsification-of-hypergraphs} holds (only) for characteristic vectors, $H$ is called an \emph{$\epsilon$-cut sparsifier} of $G$.
Although it is known that any undirected hypergraph on $n$ vertices admits cut sparsifiers of $\widetilde{O}(n^2/\epsilon^2)$ hyperedges~\cite{Kogan:2015by,Newman:2013ks}, to the best our knowledge, no spectral sparsification has been studied thus far.

For an undirected hypergraph $G=(V,E,\bmw)$, we define its \emph{size} as $\size(G)=\sum_{e \in E}|e|$.
In this work, we show that we can spectrally sparsify undirected hypergraphs with a polynomial number of hyperedges.
Note that an undirected hypergraph on $n$ vertices can have $\Omega(2^n)$ hyperedges.
\begin{theorem}\label{the:hypergraph}
  There exists a randomized algorithm that, given an undirected hypergraph $G=(V,E,\bmw)$ and $\epsilon \in (0,1)$, outputs an $\epsilon$-spectral sparsifier $H$ of $G$ with $O(n^3 \log n /\epsilon^2)$ hyperedges with probability at least $1-1/n$ in $O(pn + m\log(1/\epsilon^2)+n^3\log n/\epsilon^2 )$ time, where $n=|V|$, $m=|E|$, and $p=\size(G)$.
\end{theorem}

A directed hypergraph $G=(V,E,\bmw)$ consists of a vertex set $V$, a set of hyperarcs $E$, and a weight function $\bmw \in \bbR_+^E$, where each hyperarc $e$ is a pair $(T_e,H_e)$ of (unnecessarily disjoint) sets of vertices, where $T_e,H_e \subseteq V$ are called the \emph{tail} and \emph{head}, respectively, of $e$.
This notion was first introduced in~\cite{Gallo:1993jc}, in which applications in propositional logic, dependency analysis in relational database, and traffic analysis were considered.
A directed hypergraph coincides with an undirected hypergraph when $T_e=H_e$ for every $e \in E$ and a directed graph when $|T_e|=|H_e|=1$ for every $e \in E$.

Using the framework of the submodular Laplacian~\cite{Yoshida:2017zz}, one can derive the Laplacian $L_G\colon \bbR^V \to \bbR^V$ for a directed hypergraph $G=(V,E,\bmw)$, and its ``quadratic form'' $\bmx^\top L_G(\bmx)$ can be written as
\[
  \bmx^\top L_G(\bmx) = \sum_{e=(T_e,H_e) \in E}\bmw(e) \max\limits_{u \in T_e}\max\limits_{v \in H_e}\bigl([\bmx(u)-\bmx(v)]^+\bigr)^2,
\]
where $[x]^+=\max\set{x,0}$.
We say that a hyperarc $e$ is \emph{cut} by a vertex set $S \subseteq V$ if $T_e \cap S \neq \emptyset$ and $H_e \cap (V\setminus S) \neq \emptyset$.
As with the undirected case, the quadratic form can represent cut weights and has applications in semi-supervised learning~\cite{Zhang:2017va}.
We define $\epsilon$-spectral and $\epsilon$-cut sparsifiers in the same manner as in the undirected case.

For a directed hypergraph $G=(V,E,\bmw)$, we define its \emph{size} as  as $\size(G)=\sum_{(T_e,H_e) \in E}(|T_e|+|H_e|)$.
We show that we can spectrally sparsify directed hypergraphs with a polynomial number of hyperarcs.
\begin{theorem}\label{the:directed-hypergraph}
  There exists a randomized algorithm that, given a weighted directed hypergraph $G=(V,E,\bmw)$ and $\epsilon \in (0,1)$, outputs an $\epsilon$-spectral sparsifier $H$ of $G$ with $O(n^3 \log n /\epsilon^2)$ hyperarcs with probability at least $1-1/n$ in $O(pn + m\log(1/\epsilon^2)+n^3\log n/\epsilon^2 )$ time, where $n=|V|$, $m=|E|$, and $p=\size(G)$.
\end{theorem}
We note that even cut sparsifiers were not known for directed hypergraphs.
This result is interesting because directed graphs, which could have $\Omega(n^2)$ arcs, do not admit non-trivial cut/spectral sparsification of size $o(n^2)$~\cite{Cohen:2017fy}, whereas directed hypergraphs, which could have $O(2^{2n})=O(4^n)$ hyperarcs, admit non-trivial cut/spectral sparsification of size $\widetilde{O}(n^3)$.

\subsection{Applications}

Our sparsification results can be obviously used to improve the time and space complexities of algorithms for solving problems that involve the quadratic forms of hypergraph Laplacians, such as minimizing/maximizing cut weight subject to some constraints, computing eigenvalues, solving Laplacian systems, and semi-supervised learning.

Moreover, sparsification of directed hypergraphs is useful for obtaining a concise representation of a submodular function.
A set function $f\colon 2^V \to \bbR$ is called \emph{submodular} if $f(S)+f(T) \geq f(S \cup T) + f(S \cap T)$ for every $S,T \subseteq V$.
The cut function $\kappa_G\colon 2^V \to \bbR_+$ of a hypergraph $G$ is a well-known example of submodular functions.
Fujishige and Patkar~\cite{Fujishige2001} showed that any nonnegative submodular set function $f\colon 2^V \to \bbR_+$ with $f(\emptyset)=f(V)=0$ can be represented as a cut function of some directed hypergraph.
We immediately obtain the following by Theorem~\ref{the:directed-hypergraph}:
\begin{corollary}[Sparse approximation of submodular functions]\label{cor:approximation}
  For any nonnegative submodular function $f\colon 2^V \to \bbR_+$ with $f(\emptyset)=f(V)=0$ and any $\epsilon \in (0,1)$, there exists a directed hypergraph $G=(V,E,\bmw)$ with $O(n^3\log{n}/\epsilon^3)$ hyperarcs for $n=|V|$ such that its cut function $\kappa_G\colon 2^V \to \bbR_+$ satisfies $(1-\epsilon)\kappa_G(S) \leq f(S) \leq (1+\epsilon)\kappa_G(S)$ for every $S \subseteq V$.
\end{corollary}
We note that we need $\Omega(2^n)$ hyperarcs in general to exactly represent the original function~\cite{Fujishige2001}.
Corollary~\ref{cor:approximation} implies that a slight approximation allows us to reduce the number of hyperarcs to a polynomial number in $n$.

As an application of Corollary~\ref{cor:approximation}, we show that the class of nonnegative submodular functions $f\colon 2^V \to [0,1]$ with $f(\emptyset)=f(V)=0$ is properly and agnostically learnable with a small number of samples:
\begin{corollary}\label{cor:agnostic-learning}
  Let $\caC$ be the class of submodular functions $f\colon 2^V \to [0,1]$ with $f(\emptyset)=f(V)=0$ and let $\caD$ be any distribution on $2^V \times \bbR$.
  Then, for any $\epsilon > 0$ and $\delta \in (0,1)$, there exists an algorithm that, by drawing $O(n^4 \log (n/\epsilon)/\epsilon^4+\log(1/\delta)/\epsilon^2 )$ samples from $\caD$ for $n=|V|$, outputs a submodular function $h\colon 2^V \to \bbR_+$ that can be evaluated in polynomial time in $n$, and
  \[
    \E_{(S,b) \sim \caD}|h(S)- b| \leq \mathrm{opt}+\epsilon,
  \]
  holds with probability at least $1-\delta$,
  where $\mathrm{opt}=  \min\limits_{f \in \caC} \E\limits_{(S,b) \sim D}|f(S)-b|$.
\end{corollary}
This improves the previous best result~\cite{Cheraghchi:2012up}, which requires $n^{O(1/\epsilon^2 )} \log(1/\delta)$ samples, only works when the marginal distribution of $\caD$ over $2^V$ is a product distribution, and is improper, that is, it may output a function that is not submodular.

A drawback of the algorithm presented in Corollary~\ref{cor:agnostic-learning} is that it may require exponential time although the number of samples is a polynomial.
By contrast, the algorithm presented in~\cite{Cheraghchi:2012up} only requires $n^{O(1/\epsilon^2 )}  \log(1/\delta)$ time.
However, this is as expected, because agnostic learning of submodular functions $f\colon 2^V \to[0,1]$ with $f(\emptyset)=f(V)=0$ in $n^{o(1/\epsilon^{2/3})}$ time implies learning $k$-parities with noise in $n^{o(k)}$ time~\cite{feldman2013representation}\footnote{Although the original claim was about agnostic learning of \emph{monotone} submodular functions, it was achieved first by showing the same claim for non-monotone submodular functions with $f(\emptyset)=f(V)=0$ and using it.}, which is a long-standing open problem in learning theory.
Hence, Corollary~\ref{cor:agnostic-learning} reveals a gap between the sample complexity and the time complexity for learning submodular functions.


\subsection{Proof Ideas}

We now describe our proof ideas.
For simplicity, let us assume that the input hypergraph $G=(V,E,\bmw)$ is undirected and unweighted, that is, $\bmw(e)=1$ for every $e \in E$.
The arguments for the directed and weighted cases are more tedious but the ideas are similar.

First, we describe our construction of sparsifiers.
For vertices $u,v\in V$, let $d_G(u,v)$ be the number of hyperedges including both $u$ and $v$.
For a hyperedge $e \in E$, let $p_e = 1/\min_{u,v\in e}d_G(u,v)$.
Then, to construct an $\epsilon$-spectral sparsifier $H$ of $G$, for each hyperedge $e \in E$, we repeat the following process $K:=O(n\log n/\epsilon^2)$ times.
Add a copy of $e$ with weight $\frac{1}{Kp_e}$ to $H$ with probability $p_e$ and do nothing for the complement probability $1-p_e$.
For each $e \in E$, the expected weight of copies of $e$ in $H$ is exactly one, and the expected number of hyperedges in $H$ is $K\sum_{e \in E}p_e$.

The difficulty in the analysis arises because \eqref{eq:quadratic-form-hypergraph} involves a maximum operation, which makes it difficult to use linear algebraic tools.
A key observation for resolving this issue is that~\eqref{eq:quadratic-form-hypergraph} coincides with the quadratic form of the Laplacian associated with an undirected graph if we fix the ordering of elements in the vector $\bmx \in \bbR^V$.
To see this, assume that $V = \set{1,2,\ldots,n}$ and let $\pi$ be a permutation of $[n]$ such that $\bmx(\pi_1) \geq \dots \geq \bmx(\pi_n)$ (break ties arbitrary).
Let $G_\pi$ be the graph obtained by replacing each hyperedge $e \in E$ with an edge $\set{s_e,t_e}$, where $s_e = \argmin_{v \in e}\pi^{-1}_v$ and $t_e = \argmax_{v \in e}\pi^{-1}_v$.
In other words, $s_e$ and $t_e$ are the first and last vertices of $e$ in the ordering $\pi_1,\ldots,\pi_n$, respectively.
Then, we have
\[
  \bmx^\top L_G(\bmx) = \sum_{e \in E}\max_{u,v \in e}\bigl(\bmx(u)-\bmx(v)\bigr)^2 = \sum_{e\in E}\bigl(\bmx(s_e)-\bmx(t_e))^2 = \bmx^\top L_{G_\pi}\bmx.
\]
This means that, if $H_\pi$ is an $\epsilon$-spectral sparsifier of $G_\pi$ for \emph{every} permutation $\pi$, then $H$ is an $\epsilon$-spectral sparsifier of $G$.
It is known that, if we sample each edge $e$ of $G_\pi$ with a probability no less than a quantity called the effective resistance of $e$ and repeat this process $O(\log n/\epsilon^2)$ times, we get an $\epsilon$-spectral sparsifier of $G_\pi$ with high probability~\cite{Spielman:2011ir}.
Hence, a naive approach is sampling each hyperedge $e$ with a probability no less than the maximum effective resistance of the corresponding edges of $e$ in $G_\pi$ over the choice of $\pi$ and repeating this process $O(\log n/\epsilon^2)$ times.
However, as discussed in Appendix~\ref{sec:resistance}, this strategy may leave a hypergraph with an exponential number of hyperedges.

A crucial observation for addressing the above-mentioned issue is that, when constructing an $\epsilon$-spectral sparsifier of $G_\pi$, we need to consider vectors only in $\bbR^V_\pi := \set{\bmx \in \bbR^V \mid \bmx(\pi_1) \geq \cdots \geq \bmx(\pi_n)}$, that is, we only need to guarantee that

\[
  (1-\epsilon)\bmx^\top L_{H_\pi}\bmx \leq \bmx^\top L_{G_\pi}\bmx \leq (1+\epsilon)\bmx^\top L_{H_\pi}\bmx
\]
for any vector $\bmx \in \bbR_\pi^V$.
In what follows, we only discuss how to guarantee $(1-\epsilon)\bmx^\top L_{H_\pi}\bmx \leq \bmx^\top L_{G_\pi}\bmx$ as the other direction is similar.

Suppose that we have constructed a hypergraph $H$ based on $d_G(u,v)$ as mentioned above (that is, we set $p_e = 1/\min_{u,v \in e}d_G(u,v)$).
To show that $(1-\epsilon)\bmx^\top L_{H_\pi}\bmx \leq \bmx^\top L_{G_\pi}\bmx$ for $\bmx \in \bbR^V_\pi$, we re-parameterize $\bmx$ using the $(n-1)$-dimensional vector $\bmy = (\bmx(\pi_1)-\bmx(\pi_2),\ldots,\bmx(\pi_{n-1})-\bmx(\pi_n))$.
Let $Q_\pi \in \bbR^{(n-1) \times n}$ be the matrix such that $\bmy = Q_\pi \bmx$.
Then, it is possible to represent $L_{G_\pi} = Q_\pi^\top B_\pi^\top B_\pi Q_\pi$ and $L_{H_\pi} = Q_\pi^\top B_\pi^\top W_H B_\pi Q_\pi$ for some matrices $B_\pi \in \bbR^{E \times (n-1)}$ and $W_H \in \bbR^{E \times E}$.
Here, $B_\pi$ is a $\set{0,1}$-valued matrix with each row having an interval of consecutive ones and $W_H$ is a diagonal matrix whose $(e,e)$-th element is the total weight of copies of $e$ in $H$.
Then, it suffices to show that  $(1-\epsilon)\bmy^\top B_\pi^\top W_H B_\pi  \bmy \leq \bmy^\top B_\pi^\top B_\pi \bmy $ for every $\bmy \in \bbR_+^{n-1}$, that is, $B_\pi^\top B_\pi - (1-\epsilon)B_\pi^\top W_H B_\pi$ is \emph{copositive}.
To show that the matrix is copositive, we show by using a concentration bound that, as long as $p_e$  is no less than some threshold $p_e^\pi$, we have $B_\pi^\top B_\pi - (1-\epsilon)B_\pi^\top W_H B_\pi \geq O$ with high probability, where $O \in \bbR^{(n-1)\times (n-1)}$ is the zero matrix.
Then, we show that $\max_\pi p_e^\pi = 1/\min_{u,v\in e}d_G(u,v)$.
This implies that, from our choice of $p_e$, the resulting hypergraph $H$ acts as an $\epsilon$-spectral sparsifier $H_\pi$ of $G_\pi$ for every permutation $\pi$ with high probability, which means that $H$ is an $\epsilon$-spectral sparsifier of $G$ with high probability.
As we can show that $\sum_{e \in E }p_e =\sum_{e \in E }1/\min_{u,v \in e}d_G(u,v)= O(n^2)$ holds, $H$ has $O(n^2 K) = O(n^3\log n/\epsilon^2)$ hyperedges on average.



\subsection{Related Work}
In this section $n$, $m$, and $p$ denote the number of vertices, number of edges, and size, respectively, of the input (hyper)graph.

The first general construction of cut sparsifiers for graphs is attributed to Bencz{\'u}r and Karger~\cite{Benczur:1996fs}.
They presented a polynomial-time algorithm that constructs an $\epsilon$-cut sparsifier with $O(n\log n/\epsilon^2)$ edges.
Their sparsifier is constructed by randomly sampling each edge with probability proportional to a parameter called \emph{edge strength}.
Fung~\textit{et~al.}~\cite{Fung:2011gp} showed that we can construct an $\epsilon$-cut sparsifier with $O(n\log n/\epsilon^2)$ edges by sampling each edge with probability proportional to (an approximation to) the edge-connectivity of its end points, reducing the time complexity to $\widetilde{O}(n/\epsilon^2+m)$.

Spielman and Teng~\cite{Spielman:2011kia} introduced the concept of a spectral sparsifier for graphs and proposed a construction of a spectral sparsifiers with $O(n\log^c n)$ edges for some constant $c>0$.
The number of edges was later improved to $O(n \log n/\epsilon^2)$ by sampling every edge with probability proportional to its effective resistance~\cite{Spielman:2011ir}.
Since then, various constructions of spectral sparsifier have been proposed~\cite{AllenZhu:2015cm,Batson:2014io,Lee:2015ex} and the current best methods use merely $O(n/\epsilon^2)$ edges and run in $\widetilde{O}(m/\epsilon^c)$ time for some $c>1$~\cite{Lee:2017bu} or use $\widetilde{O}(n/\epsilon)$ edges and run in $\widetilde{O}(m)$ time~\cite{Jambulapati:2018va}.

It was implicitly shown in~\cite{Newman:2013ks} that every undirected hypergraph admits an $\epsilon$-cut sparsifier with $\widetilde{O}(n^2 \log n/\epsilon^c)$ hyperedges for some $c > 1$.
Kogan and Krauthgamer~\cite{Kogan:2015by} generalized the method of Bencz{\'u}r and Karger~\cite{Benczur:1996fs} and showed that any $r$-uniform hypergraph admits an $\epsilon$-cut sparsifier of $O(n(r+\log n)/\epsilon^2)$ hyperedges.
Chekuri and Xu~\cite{Chekuri:2017td}, based on the idea of~\cite{Benczur:2015bm}, improved the running time to $O(p\log^2 n \log p+nr(r+\log n)/\epsilon^2)$, although the resulting cut sparsifier has $O(nr(r+\log n)/\epsilon^2)$ hyperedges, which is $r$ times greater than the size of~\cite{Kogan:2015by}.
Theorem~\ref{the:hypergraph} extends these results to spectral sparsifiers and Theorem~\ref{the:directed-hypergraph} further extends them to directed hypergraphs.



Several variants of Laplacians for hypergraphs have been proposed in the literature besides that the one used in this work.
In~\cite{Rodriguez2002}, the cut value of a set $S \subseteq V$ was defined as $\sum_{e \in E} \bmw(e)|S \cap e||e \setminus S|$.
De~Carli~Silva~\emph{et~al.}~\cite{Silva2015} considered a hypergraph Laplacian by extending this cut function and showed that every hypergraph admits a spectral sparsifier of $O(n/\epsilon^2)$ hyperedges in this sense.
Although they also provided a construction of cut sparsifiers for $r$-uniform hypergraphs with $O(n/\epsilon)$ edges using our definition of cut, their guarantee is merely $\frac{4(r-1)}{r^2} \bmone_S L_H(\bmone_S) \leq \bmone_S L_G(\bmone_S) \leq \frac{(1+\epsilon)r^2}{4(r-1)}\bmone_S L_H(\bmone_S)$ for every $S \subseteq V$, which is significantly weaker than the guarantee of $\epsilon$-cut sparsifiers.

Learning of submodular functions with additive error was first considered by Gupta~\emph{et~al.}~\cite{gupta2013privately}, who showed that we can learn $[0,1]$-valued submodular functions with $n^{O(1/\epsilon^2)}$ value queries and applied the result to private data release.
Cheraghchi~\emph{et~al.}~\cite{Cheraghchi:2012up} showed that submodular functions are noise stable, which leads to agnostic learning of $[0,1]$-valued submodular functions with $n^{O(1/\epsilon^2)}$ samples from a product distribution.
Subsequently, Feldman~\emph{et~al.}~\cite{feldman2013representation} showed that PAC learning of submodular functions can be done in $\mathrm{poly}(n)\cdot 2^{O(1/\epsilon^4)}$ time for any distribution and that it requires $2^{\Omega(1/\epsilon^{2/3})}$ samples (or even value queries).
As we mentioned previously, they also presented evidence that agnostic learning of submodular functions requires $n^{\Omega(1/\epsilon^{2/3})}$ time.
Raskhodnikova and Yaroslavtsev~\cite{raskhodnikova2013learning} considered learning submodular functions taking values in the range $\set{0,1,\ldots,k}$.
They built up on the approach of~\cite{gupta2013privately} to show that such submodular functions can be expressed as $2k$-DNF and then applied Mansour’s algorithm for learning DNF~\cite{mansour1995nlog} to obtain a $\mathrm{poly}(n) \cdot k^{O(k \log k/\epsilon)}$-time PAC learning algorithm using value queries.

\subsection{Organization}
The rest of this paper is organized as follows.
Section~\ref{sec:pre} introduces basic concepts for hypergraphs and technical tools.
In Section~\ref{sec:order}, we relate the sparsification for vectors in $\bbR_\pi^V$ for a permutation $\pi$ and the cone of copositive matrices.
Section~\ref{sec:graph} describes spectral sparsification of \emph{graphs} for vectors in $\bbR_\pi^V$.
Section~\ref{sec:hypergraph} describes and analyzes our spectral sparsification of hypergraphs.
Finally, Section~\ref{sec:application} discusses some applications of spectral sparsification of hypergraphs.


\section{Preliminaries}\label{sec:pre}

For a positive integer $n \in \bbN$, we denote the set $\set{1,2,\ldots,n}$ by $[n]$.
For a vector $\bmx \in \bbR^n$, we define $\diag(\bmx) \in \bbR^{n \times n}$ as the diagonal matrix whose $(i,i)$-th entry is $\bmx(i)$ for every $i \in [n]$.
For a probabilistic event $X$, let $[X] \in \set{0,1}$ be the indicator of $X$.

Let $G=(V,E,\bmw)$ be a hypergraph.
The \emph{size} of $G$, denoted by $\size(G)$, is $\sum_{e \in E}|e|$ if $G$ is undirected and is $\sum_{(T_e,H_e)\in E}|T_e|+|H_e|$ if $G$ is directed.
For a set of hyperedges/hyperarcs $F\subseteq E$, we denote the subgraph $(V,E\setminus F,\bmw|_{E \setminus F})$ by $G\setminus F$, where $\bmw|_{E \setminus F} \in \bbR_+^{E \setminus F}$ is the restriction of $\bmw$ to $E \setminus F$.


The Moore-Penrose inverse of a matrix $A$ is denoted by $A^\dagger$.
The sets of $n \times n$ symmetric matrices and positive semidefinite matrices are denoted by $\bbS^n$ and $\bbS_+^n$, respectively.
A symmetric matrix $A \in \bbS^{n}$ is called \emph{copositive} if $\bmx^\top A \bmx \geq 0$ for any $\bmx \in \bbR_+^n$.
For two symmetric matrices $A,B \in \bbS^n$, we write $A \succeq B$ if $A-B$ is positive semidefinite.
For a symmetric matrix $A \in \bbS^{n}$, $\lambda_{\min}(A)$ and $\lambda_{\max}(A)$ denote the smallest and largest eigenvalues of $A$, respectively.

A convex cone $C$ is said to be \emph{pointed} if it does not contain a line, or equivalently, $\bmx \in C$ and $-\bmx \in C$ implies $\bmx = \bmzero$.

\begin{lemma}[Chernoff's bound]\label{lem:chernoff}
  Let $X_1, \ldots, X_n$ be independent random variables bounded by the interval $[0, R]$.
  Then, for $X = X_1 + \cdots + X_n$ and $\epsilon\in(0,1)$, we have
  \[
    \Pr \bigl[|X- \mu|\geq \epsilon \mu \bigr]\leq 2\exp\left(-\frac{\epsilon^2 \mu}{3R}\right).
  \]
  where $\mu = \E[X]$.
\end{lemma}


Let $\pi$ be a permutation of $[n]$.
Then, let $\bbR^n_\pi$ be the set $\set{\bmx \in \bbR^n \mid \bmx(\pi_1) \geq \cdots \geq \bmx(\pi_n)}$ and let $P_\pi \in \bbR^{n\times n}$ be the permutation matrix associated with $\pi$, that is,
\[
  P_\pi(i,j) = \begin{cases}
    1 & \text{if }j = \pi_i, \\
    0 & \text{otherwise}.
  \end{cases}
\]


\section{Cone of $\bbR_+^\pi$-Copositive Matrices}\label{sec:order}
For a matrix $A \in \bbS^n$ and a permutation $\pi$, we say that $A$ is \emph{$\bbR_+^\pi$-copositive} if $\bmx^\top A \bmx \geq 0$ for any $\bmx \in \bbR_\pi^n$.
In this section, we characterize $\bbR_+^\pi$-copositive matrices $A \in \bbS^{n}$ with $A\bmone = \bmzero$ using copositive matrices.


\begin{proposition}
  Let $\pi$ be a permutation of $[n]$.
  The set $\caC$ of symmetric $\bbR_+^\pi$-copositive matrices is a proper cone, that is, a pointed, closed, and full-dimensional convex cone.
\end{proposition}
\begin{proof}
    The closedness and convexity are trivial.
    The full-dimensionality follows since $\caC$ contains the positive semidefinite cone, which is full-dimensional.
    For the pointedness, assume that $A, -A \in \caC$.
    This implies that $\bmx^\top A \bmx = 0$ for any $\bmx\in\bbR^n_\pi$.
    Since $\dim \bbR^n_\pi = n$, $A = O$.
\end{proof}

For two symmetric matrices $A,B \in \bbS^{n}$ and a permutation $\pi$ of $[n]$, we define the relation $A \succeq_\pi B$ as $\bmx^\top (A-B)\bmx \geq 0$ for every $\bmx \in \bbR^n_\pi$.
Note that $A$ is $\bbR_+^\pi$-copositive if and only if $A \succeq_\pi O$.
It is well known that any proper cone $\caC$ induces a partial order $\succeq_\caC$ by defining $x \succeq_\caC y$ as $x - y \in \caC$ (see~\cite[Section~2.4]{Boyd2004convex}).
The following is immediate from this fact.
\begin{proposition}
  For any permutation $\pi$ of $[n]$, the relation $\succeq_\pi$ is a partial order on $\bbS^n$.
\end{proposition}





To characterize $\bbR_+^\pi$-copositive matrices, we introduce a matrix $J \in \bbR^{(n-1) \times n}$ defined as
\[
  J(i,j) = \begin{cases}
    1  & \text{if }j = i   \\
    -1 & \text{if }j = i+1 \\
    0  & \text{otherwise}.
  \end{cases}
  \quad \text{that is,} \quad
  J =
  \begin{bmatrix}
    1 & -1 &        &        &   \\
      & 1  & -1     &        &   \\
      &    & \ddots & \ddots &   \\
      &    &        &  1     & -1
  \end{bmatrix}.
\]
\begin{lemma}\label{lem:copositive-equivalence}
  Let $A \in \bbS^{n}$ be a matrix with $A\bmone=\bmzero$ and $\pi$ be a permutation of $[n]$.
  Then, $A$ is $\bbR_+^\pi$-copositive if and only if $A$ is of the form $P_\pi^\top J^\top C J P_\pi$ for some copositive matrix $C \in \bbR^{(n-1)\times (n-1)}$.
\end{lemma}
\begin{proof}
    For simplicity, we show the lemma for $\pi = \mathrm{id}$.
    The general case is obtained by permuting variables with $P_\pi$.

    For the sufficiency, since for any $\bmx \in \bbR^n_\pi$, $\bmx' := J\bmx \in \bbR_+^{n-1}$, we obtain $\bmx^\top J^\top C J \bmx = \bmx' C \bmx' \geq 0$.

    For the necessity, let us consider an upper triangular matrix $E \in \bbR^{n \times n}$ defined as
    \[
        E =
        \begin{bmatrix}
            1 & \cdots & 1 \\
             & \ddots &  \vdots \\
              &        & 1
        \end{bmatrix}.
    \]
    Then $E$ is a bijection from $\bbR_+^n$ to $\bbR_\pi^n$.
    For $\bmx \in \bbR_\pi^n$, defining $\bmy \in \bbR^n_+$ such that $\bmx = E\bmy$, we obtain $\bmy E^\top A E \bmy = \bmx A \bmx \geq 0$.
    Therefore $E^\top A E$ is copositive.
    Since $A\bmone = \bmzero$ and the last column of $E$ is $\bmone$, we have
    \[
    E^\top AE =
    \begin{bmatrix}
        C & \bmzero \\
        \bmzero^\top & 0
    \end{bmatrix}
\]
    for some $(n-1)\times(n-1)$ copositive matrix $C$.
    From
    $
    E^{-1} = \left[ \begin{smallmatrix} J \\ \bme_n^\top \end{smallmatrix}\right]
    $, we conclude that
    \[
        A =
        \begin{bmatrix} J^\top & \bme_n \end{bmatrix}
    \begin{bmatrix} C & \bmzero \\ \bmzero^\top & 0 \end{bmatrix}
        \begin{bmatrix} J \\ \bme_n^\top \end{bmatrix}
        = J^\top C J. \qedhere
\]
\end{proof}


\section{Spectral Sparsification of Graphs for Vectors in $\bbR^n_\pi$}\label{sec:graph}
Let $G$ be a (weighted) graph on the vertex set $V:=[n]$ and $\pi$ be a permutation of $[n]$.
We say that a weighted subgraph $H$ of $G$ is an \emph{$\epsilon$-spectral sparsifier of $G$ for vectors in $\bbR^n_\pi$} if $(1- \epsilon) \bmx^\top L_G \bmx \leq \bmx^\top L_H \bmx \leq (1+ \epsilon) \bmx^\top L_G \bmx$ holds for every $\bmx \in \bbR^n_\pi$, or more simply, $(1-\epsilon)L_H \preceq_\pi L_G \preceq_\pi (1-\epsilon)L_H$ holds.
In this section, we present a randomized algorithm that constructs  spectral sparsifiers for vectors in $\bbR^n_\pi$.

\begin{algorithm}[t!]
  \caption{}\label{alg:graph}
  \begin{algorithmic}[1]
    \Require{A graph $G=(V,E,\bmw_G)$, a permutation $\pi$, $\epsilon, \delta\in(0,1)$, and $\set{p_e}_{e \in E}$.}
    \State{$K \leftarrow \Theta(\log (n/\delta)/\epsilon^2)$.}
    \State{$\bmw_H \leftarrow \bmzero$.}
    \State{Let $(Z_{k,e})_{k \in[K], e \in E}$ be mutually independent random variables in $\set{0, 1}$ with $\E[ Z_{k,e}] = p_e$.}
    \For{$k \leftarrow 1$ to $K$}
    \For{each $e \in E$}
    \State Increase $\bmw_H(e)$ by $Z_{k,e}\bmw_G(e)/(K p_e)$.
    \EndFor
    \EndFor
    \State{\Return a graph $H=(V,E,\bmw_H)$.}
  \end{algorithmic}
\end{algorithm}

Our algorithm is a very simple sampling algorithm (Algorithm~\ref{alg:graph}).
Given a graph $G=(V,E,\bmw_G)$ and edge sampling probabilities $\set{p_e}_{e \in E}$, we repeat the following process $K:=\Theta(\log(n/\delta)/\epsilon^2)$ times:
For each edge $e \in E$, we increase the weight $\bmw_H(e)$ by $\bmw_G(e)/(Kp_e)$ with probability $p_e$.
Then, we output the graph $H=(V,E,\bmw_H)$.
Clearly, the expected number of edges with non-zero weights in the output graph is $O\bigl(K \sum_{e \in E}p_e\bigr) = O\bigl(\sum_{e \in E}p_e \log (n/\delta)/\epsilon^2\bigr)$, which is small if $\sum_{e \in E}p_e$ is small.

We will show that Algorithm~\ref{alg:graph} produces a spectral sparsifier with high probability if the sampling probabilities are higher than some thresholds.
To describe the thresholds, we introduce some definitions.
For $i \in [n-1]$, we say that an edge $e=\set{u,v} \in E$ \emph{crosses $i$ in the ordering induced by $\pi$} if $\min\set{\pi^{-1}_u,\pi^{-1}_v} \leq i < \max\set{\pi^{-1}_u,\pi^{-1}_v}$, and we denote it by $i \in_\pi e$.
In words, when we align the vertices of $G$ as $\pi_1,\ldots,\pi_n$ on a line in this order, the edge $\set{u,v}$ crosses the bisector between $\pi_i$ and $\pi_{i+1}$ (see Figure~\ref{fig:cross}).
For $i,j \in [n-1]$, let $d_{G,\pi}(i,j)$ be the total weight of edges that cross both $i$ and $j$ in $\pi$.
Note that $d_{G,\pi}(i,j)=d_{G,\pi}(j,i)$.

\begin{figure}[t]
    \center
    \includegraphics[]{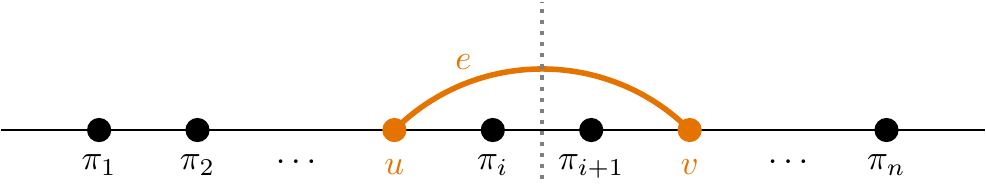}
    \caption{An edge crossing $i$ in the order induced by $\pi$.}\label{fig:cross}
\end{figure}

In the rest of this section, we prove the following:
\begin{theorem}\label{the:graph}
  Suppose that
  \[
    p_e \geq \frac{\bmw_G(e)}{\min_{i,j \in_\pi e}d_{G,\pi}(i,j)}
  \]
  holds for every $e \in E$.
  Then, the graph $H$ produced by Algorithm~\ref{alg:graph} is an $\epsilon$-spectral sparsifier for vectors in $\bbR^n_\pi$ with probability at least $1-\delta$.
\end{theorem}

Our goal is to show that $(1-\epsilon)L_H \preceq_\pi L_G \preceq_\pi (1+\epsilon)L_H$ under the assumption of Theorem~\ref{the:graph}.
To this end, we rephrase this condition using Lemma~\ref{lem:copositive-equivalence}.
Let $W_G,W_H \in \bbR^{E \times E}$ be the diagonal matrices whose $(e,e)$-th entries are $\bmw_G(e)$ and $\bmw_H(e)$, respectively, for each $e \in E$.
Then, we can write $L_G=B_G^\top W_G B_G$ and $L_H=B_G^\top W_H B_G$, where $B_G \in \bbR^{E \times V}$ is the incidence matrix of $G$.
The following is a simple but important observation:
\begin{proposition}
  We can write $B_G = B_\pi J P_\pi$, where $B_\pi \in \bbR^{E \times (n-1)}$ is a matrix defined as
  \[
    B_\pi(e,i) = \begin{cases}
      1 & \text{if }i \in_\pi e, \\
      0 & \text{otherwise}.
    \end{cases}
  \]
\end{proposition}
Note that each row of $B_\pi$ has consecutive ones.
Then, Lemma~\ref{lem:copositive-equivalence} implies that $(1-\epsilon)L_H \preceq_\pi L_G \preceq_\pi (1+\epsilon)L_H$ holds if and only if $(1+\epsilon)B_\pi^\top W_H B_\pi - B_\pi^\top W_G B_\pi$ and $B_\pi^\top W_G B_\pi - (1-\epsilon)B_\pi^\top W_H B_\pi$ are copositive.
The elements of these matrices can be written as follows:
\begin{lemma}\label{lem:element}
  For $i,j \in [n-1]$, we have $(B_\pi^\top W_G B_\pi)(i,j) = d_{G,\pi}(i,j)$ and $(B_\pi^\top W_H B_\pi)(i,j) = d_{H,\pi}(i,j)$.
\end{lemma}
\begin{proof}
  We have
  \[
    (B_\pi^\top W_G B_\pi)(i,j) = \sum_{e \in E}\bmw_G(e) [i \in_\pi e \wedge j \in_\pi e] = d_{G,\pi}(i,j).
  \]
  The same argument holds for $(B_\pi^\top W_H B_\pi)(i,j)$.
\end{proof}

\begin{proof}[Proof of Theorem~\ref{the:graph}]
  To show the copositivity of $(1+\epsilon)B_\pi^\top W_H B_\pi - B_\pi^\top W_G B_\pi$ and $B_\pi^\top W_G B_\pi - (1-\epsilon)B_\pi^\top W_H B_\pi$, it suffices to show that they are non-negative matrices with probability at least $1-\delta$.
  To this end, we show that, for every $i,j \in [n-1]$, $(1-\epsilon)d_{H,\pi}(i,j) \leq d_{G,\pi}(i,j) \leq (1+\epsilon)d_{G,\pi}(i,j)$ with probability at least $1-\delta/\binom{n}{2}$.
  Then, a union bound over the choice of $i,j \in [n-1]$ and Lemma~\ref{lem:element} gives the desired result.


  Fix $i,j \in [n-1]$.
  Note that $d_{H,\pi}(i,j) = \sum_{k,e}X_{k,e}$, where
  \[
    X_{k,e} =\frac{Z_{k,e}\bmw_G(e)}{K p_e} [i,j\in_\pi e]. \label{eq:X-i-e}
  \]
  Then, we have
  \[
    \E\bigl[d_{H,\pi}(i,j)\bigr] = K \sum_{e \in E} \frac{\E[Z_{k,e}]\bmw_G(e)}{K p_e}[i,j\in_\pi e] = \sum_{e \in E} \bmw_G(e)[i,j\in_\pi e] = d_{G,\pi}(i,j).
  \]
  In addition, from the assumption on $\set{p_e}_{e \in E}$, we have
  \[
    X_{k,e} \leq \frac{\min_{i',j'\in_\pi e}d_{G,\pi}(i',j')}{K}[i,j\in_\pi e]\leq \frac{d_{G,\pi}(i,j)}{K}.
  \]
  By Chernoff's bound (Lemma~\ref{lem:chernoff}), we have
  \[
    \Pr\Bigl[ |d_{H,\pi}(i,j) - d_{G,\pi}(i,j)| \geq \epsilon d_{G,\pi}(i,j) \Bigr]  \leq 2\exp\Bigl(-\frac{\epsilon^2 K }{3}\Bigr)
    \leq
    \frac{\delta}{\binom{n}{2}}.
  \]
  by setting the hidden constant in $K$ to be sufficiently large.
\end{proof}


\section{Spectral Sparsification of Hypergraphs}\label{sec:hypergraph}

In this section, we prove Theorems~\ref{the:hypergraph} and~\ref{the:directed-hypergraph}.

Our algorithm, given in Algorithm~\ref{alg:hypergraph}, is almost identical to Algorithm~\ref{alg:graph} except that the input and the output are now hypergraphs and that the error probability $\delta$ is divided by $n!$, and it works for both the undirected and directed cases by changing the sampling probabilities, as discussed in Sections~\ref{subsec:undirected} and~\ref{subsec:directed}, respectively.

\begin{algorithm}[t!]
  \caption{}\label{alg:hypergraph}
  \begin{algorithmic}[1]
    \Require{An undirected/directed hypergraph $G=(V,E,\bmw_G)$, a permutation $\pi$, $\epsilon,\delta\in(0,1)$, and $\set{p_e}_{e \in E}$.}
    \State{$K \leftarrow \Theta(\log (n!/\delta)/\epsilon^2)$, where $n=|V|$.}
    \State{$\bmw_H \leftarrow \bmzero$.}
    \State{Let $(Z_{k,e})_{k\in[K], e \in E}$ be mutually independent random variables in $\set{0, 1}$ with $\E[ Z_{k,e}] = p_e$.}
    \For{$k \leftarrow 1$ to $K$}\label{line:hypergraph-loop-start}
    \For{each $e \in E$}
    \State Increase $\bmw_H(e)$ by $Z_{k,e}\bmw_G(e)/(K p_e)$.
    \EndFor
    \EndFor\label{line:hypergraph-loop-end}
    \State{\Return a hypergraph $H=(V,E,\bmw_H)$.}
  \end{algorithmic}
\end{algorithm}

\subsection{Undirected Hypergraphs}\label{subsec:undirected}

For a hypergraph $G=(V,E,\bmw_G)$ and two vertices $u,v\in V$, we define $d_G(u,v)$ as the total weight of hyperedges that include both $u$ and $v$.
Note that $d_G(v,u)=d_G(u,v)$.
The following lemma gives a condition on the sampling probability of the resulting hypergraph being a spectral sparsifier.
\begin{lemma}\label{lem:hypergraph-condition}
  Suppose that
  \[
    p_e \geq \frac{\bmw_G(e)}{\min\limits_{u,v\in e}d_G(u,v)}
  \]
  holds for every $e \in E$.
  Then, Algorithm~\ref{alg:hypergraph} produces an $\epsilon$-spectral sparsifier of $G$ with probability at least $1-\delta$.
\end{lemma}
\begin{proof}
  For a permutation $\pi$, let $G_\pi$ be the graph obtained from $G$ by replacing each hyperedge $e \in E$ with an edge $e_\pi := \set{s_e,t_e}$, where $s_e = \argmin_{v\in e} \pi^{-1}_v$ and $t_e = \argmax_{v\in e} \pi^{-1}_v$.
  Note that, for every $\bmx \in \bbR_\pi^V$, we have
  \[
    \bmx^\top L_G(\bmx)
    =
    \sum_{e \in E}\bmw_G(e)\max_{u,v \in e}\bigl(\bmx(u)-\bmx(v)\bigr)^2
    =
    \sum_{e \in E}\bmw_G(e)\bigl(\bmx(s_e)-\bmx(t_e)\bigr)^2
    =
    \bmx^\top L_{G_\pi}\bmx.
  \]
  For any $e \in E$, we have
  \[
    \min_\pi \min_{i,j \in_\pi e}d_{G_\pi,\pi}(i,j)
    =
    \min_{\pi:1,n-1\in_\pi e} d_{G_\pi,\pi}(1,n-1)
    =
    \min_{u,v\in e}d_G(u,v).
  \]
  (See Section~\ref{sec:graph} for the definition of $d_{G_\pi,\pi}(i,j)$.)
  Then, the claim holds by Theorem~\ref{the:graph} and a union bound over the choice of $\pi$.
\end{proof}

The following is useful for providing an upper bound on the number of hyperedges in $H$.
\begin{lemma}\label{lem:hypergraph-total-sum}
  We have
  \[
    \sum_{e \in E}\frac{\bmw_G(e)}{\min\limits_{u,v \in e}d_G(u,v)} = O(n^2),
  \]
  where $n=|V|$.
\end{lemma}
\begin{proof}
  Let $\set{u_1,v_1},\ldots,\set{u_m,v_m}$ be a sequence of (unordered) pairs of vertices in increasing order of $d_G(u_i,v_i)$, where $m = {n \choose 2}$.
  Then, consider the following iterative algorithm with $m$ steps.
  At step $i$, for each (remaining) hyperedge $e$ including both $u_i$ and $v_i$, we assign a cost of $\bmw_G(e)/d_G(u_i,v_i)$ to $e$ and delete $e$.

  The cost assigned to a hyperedge $e$ is exactly $\bmw_G(e)/\min_{u,v\in e}d_G(u,v)$;
  hence we want to upper-bound the total cost assigned to the hyperedges.
  This can be bounded by $O(n^2)$ because the total cost assigned to the hyperedges at step $i$ is at most $1$ and the number of steps is $m=O(n^2)$.
\end{proof}
We note that the bound $O(n^2)$ in Lemma~\ref{lem:hypergraph-total-sum} is tight because an unweighted complete graph satisfies $\sum_{e \in E}\bmw_G(e)/\min_{u,v\in e}d_G(u,v) = \sum_{e\in E}1= {n \choose 2} = \Omega(n^2)$.

\begin{proof}[Proof of Theorem~\ref{the:hypergraph}]
  Our algorithm first calculates $d_G(u,v)$ for each (unordered) pair $\set{u,v} \in {V \choose 2}$, and then calls Algorithm~\ref{alg:hypergraph} with $p_e = \bmw_G(e)/\min_{u,v \in e}d_G(u,v)$ and $\delta=1/2n$.
  By Lemma~\ref{lem:hypergraph-condition}, the output is an $\epsilon$-spectral sparsifier with probability at least $1-1/2n$ by setting the hidden constant in $K$ to be sufficiently large.
  The expected number of hyperedges in the output is $O(K\sum_{e \in E}p_e) = O(n^3\log n/\epsilon^2)$ by Lemma~\ref{lem:hypergraph-total-sum}.
  Then, our algorithm outputs an $\epsilon$-spectral sparsifier of $O(n^3\log n/\epsilon^2)$ hyperedges with probability at least $1-1/n$.

  Now, we analyze the time complexity.
  To compute $d_G(u,v)$'s, we introduce a counter $c_{uv}$ initialized to be zero for each unordered pair $\set{u,v}$ of vertices.
  Then for each hyperedge $e \in E$, we increase the counter $c_{uv}$ by $\bmw_G(e)$ for every unordered pair $\set{u,v} \subseteq e$.
  The time complexity for this part is $\sum_{e \in E}|e|^2 \leq \sum_{e \in E}|e|n=pn$.
  To efficiently simulate the process from Line~\ref{line:hypergraph-loop-start} to Line~\ref{line:hypergraph-loop-end} in Algorithm~\ref{alg:hypergraph}, for each hyperedge $e \in E$, we set $\bmw_H(e)$ to be $\frac{\bmw_G(e)X_e}{Kp_e}$, where $X_e$ is sampled from the binomial distribution with a success probability $p_e$ and the number of trials $K$, which can be done in $O(\log K)=O(\log (n\log n/\epsilon^2))$ time~\cite{Devroye:1986zz}.
  Hence, the total time complexity is $O(pn + m \log (n \log n/\epsilon^2) + n^3 \log n/\epsilon^2) = O(pn + m\log(1/\epsilon^2)+n^3\log n/\epsilon^2)$.
\end{proof}

\subsection{Directed Hypergraphs}\label{subsec:directed}
The analysis for directed hypergraphs is almost identical to that for undirected hypergraphs.
For a directed hypergraph $G=(V,E,\bmw)$ and two vertices $u,v\in V$, we define $d_G(u,v)$ as the total weight of hyperarcs $e=(T_e,H_e)$ with $u \in T_e$ and $v \in H_e$.
Note that, in this case, $d_G(u,v) \neq d_G(v,u)$ in general.
The following lemma is analogous to Lemma~\ref{lem:hypergraph-condition}.
\begin{lemma}\label{lem:directed-hypergraph-condition}
  Suppose that
  \[
    p_e \geq \frac{\bmw_G(e)}{\min\limits_{u \in T_e}\min\limits_{v \in H_e}d_G(u,v)}
  \]
  holds for every $e =(T_e,H_e) \in E$.
  Then, Algorithm~\ref{alg:hypergraph} produces an $\epsilon$-spectral sparsifier of $G$ with probability at least $1-\delta$.
\end{lemma}
\begin{proof}
  For a permutation $\pi$, let $G_\pi$ be the undirected graph obtained from $G$ by replacing each hyperarc $e=(T_e,H_e) \in E$ with an (undirected) edge $e_\pi := \set{s,t}$, where $s_e = \argmin_{u\in T_e} \pi^{-1}_u$ and $t_e = \argmax_{v\in H_e} \pi^{-1}_v$ if $\pi^{-1}_{s_e} < \pi^{-1}_{t_e}$,  and deleting it otherwise.
  Note that, for every $\bmx \in \bbR_\pi^V$, we have
  \begin{align*}
    & \bmx^\top L_G(\bmx)
    =
    \sum_{e \in E}\bmw_G(e) \max_{u \in T_e}\max_{v \in H_e}\bigl([\bmx(u)-\bmx(v)]^+\bigr)^2 \\
    & =
    \sum_{e \in E: \pi^{-1}_{s_e} < \pi^{-1}_{t_e}}\bmw_G(e)\bigl(\bmx(s_e)-\bmx(t_e)\bigr)^2
    =
    \bmx^\top L_{G_\pi}\bmx.
  \end{align*}
  For any $e \in E$, we have
  \[
    \min_\pi \min_{i,j \in_\pi e}d_{G_\pi,\pi}(i,j)
    =
    \min_{\pi:1,n-1\in_\pi e} d_{G_\pi,\pi}(1,n-1)
    =
    \min_{u,v\in e}d_G(u,v).
  \]
  Then, the claim holds by Theorem~\ref{the:graph} and a union bound over the choice of $\pi$.
\end{proof}

The following lemma is analogous to Lemma~\ref{lem:hypergraph-total-sum}
\begin{lemma}\label{lem:directed-hypergraph-total-sum}
  We have
  \[
    \sum_{e \in E}\frac{\bmw_G(e)}{\min\limits_{u \in T_e}\min\limits_{v \in H_e}d_G(u,v)} = O(n^2).
  \]
\end{lemma}
\begin{proof}
  Let $(u_1,v_1),\ldots,(u_m,v_m)$ be a sequence of (ordered) pairs of vertices in increasing order of $d_G(u_i,v_i)$, where $m = n(n-1)$.
  Then, consider the following iterative algorithm with $m$ steps.
  At step $i$, for each (remaining) hyperarc $e=(T_e,H_e)\in E$ with $u_i \in T_e$ and $v_i \in H_e$, we assign a cost of $\bmw_G(e)/d_G(u_i,v_i)$ to $e$ and delete $e$.

  The cost assigned to a hyperarc $e=(T_e,H_e)$ is exactly $\bmw_G(e)/\min_{u \in T_e}\min_{v \in H_e}d_G(u,v)$;
  hence we want to upper-bound the total cost assigned to the hyperarcs.
  This can be bounded by $O(n^2)$ because the total cost assigned to the hyperarcs at step $i$ is at most $1$ and the number of steps is $m=O(n^2)$.
\end{proof}

The proof of Theorem~\ref{the:directed-hypergraph} is identical to that of Theorem~\ref{the:hypergraph}, where we use Lemmas~\ref{lem:directed-hypergraph-condition} and~\ref{lem:directed-hypergraph-total-sum} instead of Lemmas~\ref{lem:hypergraph-condition} and~\ref{lem:hypergraph-total-sum}.

\section{Applications}\label{sec:application}
In this section, we discuss applications of spectral sparsification of hypergraphs.

\subsection{Agnostic Learning of Submodular Functions}

We now discuss an application of Corollary~\ref{cor:approximation} to agnostic learning.

We first define agnostic learning for the case in which the domain is $2^V$.
Let $\caD$ be an arbitrary distribution on $2^V \times [0,1]$.
Let $\caC$ be a class of $[0,1]$-valued functions on $2^V$.
Define the \emph{error} of $f\colon 2^V \to [0,1]$ and the \emph{optimal error} of $\caC$ as
\[
  \mathrm{err}(f) = \E_{(S,b) \sim \caD}|f(S) - b|, \quad \mathrm{opt} = \min_{f\in \caC} \mathrm{err}(f),
\]
respectively.
Roughly speaking, the objective of agnostic learning of a concept class $\caC$ is to find a hypothesis with an error not much larger than the optimal error by drawing a small number of samples from $\caD$.
Formally, we define agnostic learning as follows:
\begin{definition}[Agnostic learning]
  A concept class $\caC$ is said to be \emph{agnostically learnable} if, for any $\epsilon > 0$ and $\delta \in (0,1)$, there exists an algorithm that, given a sampling oracle from $\caD$, outputs a hypothesis $h\colon 2^V \to [0,1]$ with probability at least $1 - \delta$ such that $\mathrm{err}(h) \leq \mathrm{opt} + \epsilon$.
\end{definition}
The following is an easy consequence of Chernoff's bound and the union bound.
\begin{lemma}\label{lem:agnostic-learning-trivial}
  A concept class $\caC$ of $[0,1]$-valued functions is agnostically learnable with $O\bigl(\log (|\caC|/\delta)/\epsilon^2\bigr)$ samples.
\end{lemma}

\begin{proof}[Proof of Corollary~\ref{cor:agnostic-learning}]
  Let $f\colon 2^V \to \bbR$ be a nonnegative submodular function with $f(\emptyset)=f(V)=0$.
  We first show that $f$ can be well approximated by a sparse directed hypergraph with small weights.
  Let $G=(V,E_G,\bmw_G)$ be a directed hypergraph such that $f(S)=\kappa_G(S)$ for every $S \subseteq V$ for the cut function $\kappa_G\colon 2^V \to \bbR_+$ of $G$, whose existence is guaranteed by Corollary~\ref{cor:approximation}.
  From the proof of Corollary~\ref{cor:approximation} in~\cite{Fujishige2001}, we can observe that $0 \leq \bmw_G(e) \leq \sum_{v \in V}f(v) \leq n$.
  Let $M=O(n^3\log n/\epsilon^2)$ be the upper bound on the expected number of hyperarcs in Theorem~\ref{the:directed-hypergraph}.
  By Theorem~\ref{the:directed-hypergraph}, there exists a directed hypergraph $H=(V,E_H,\bmw_H)$ of at most $M$ hyperarcs such that $(1-\epsilon)\kappa_H(S) \leq \kappa_G(S) \leq (1+\epsilon)\kappa_H(S)$ for every $S \subseteq V$.
  Moreover, from the construction of $H$ (Algorithm~\ref{alg:hypergraph}), we know that $E_H \subseteq E_G$ and
  \[
      \bmw_H(e) \leq \frac{\bmw_G(e)}{p_e} \leq n \cdot n^2 = n^3
  \]
  for every $e \in E_H$, where in the second inequality we used $\bmw_G(e) \leq n$ and Lemma~\ref{lem:directed-hypergraph-total-sum}.

  Let $\caC$ be the class of submodular functions $f\colon 2^V \to [0,1]$ with $f(\emptyset)=f(V)=0$ and let $\caC'$ be the class of cut functions of directed hypergraphs with at most $M$ hyperarcs and each arc weight being a multiple of $\epsilon/M$.
  Then, for every $f \in \caC$, there exists $g \in \caC'$ such that $|f(S)-g(S)| \leq \epsilon$ for every $S \subseteq V$.
  Hence, it suffices to show that $\caC'$ is agnostically learnable.
  Note that we have
  \[
    |\caC'| \leq O\left(\sum_{k =0}^M \binom{2^{2n}}{k} \left(\frac{Mn^3}{\epsilon}\right)^k\right),
  \]
  which implies that
  \[
    \log |\caC'| = O\left(M \log 2^{2n}+M\log \frac{Mn^3}{\epsilon}\right) = O\left(\frac{n^4}{\epsilon^2}\log \frac{n}{\epsilon}\right).
  \]
  The claim follows by Lemma~\ref{lem:agnostic-learning-trivial}.
\end{proof}

\subsection{Effective Resistance of Hypergraphs}
As with graphs, we can define the effective resistance between a pair of vertices in a hypergraph.
For $v \in V$, let $\bme_v \in \bbR^V$ be the unit vector with $\bme_v(v)=1$ and $\bme_v(w)=0$ for $w \neq v$.
\begin{definition}[Effective Resistance~\cite{Fujii2018:submodLxb}]
  Let $G=(V,E,\bmw)$ be an undirected/directed hypergraph.
  For distinct vertices $s, t \in V$, the \emph{effective resistance} $R_G(s,t)$ between $s$ and $t$ is the maximum value of
  \begin{align}\label{eq:effres-1}
    2(\bme_s - \bme_t)^\top \bmx - \bmx^\top L_G(\bmx)
  \end{align}
  for $\bmx \in \bbR^V$.
\end{definition}

The effective resistance has the following physics interpretation.
Let us regard a hypergraph as an electric circuit, where in the undirected case, each hyperedge $e \in E$ acts as an imaginary unit such that the current flows from the vertex in $e$ with the highest potential to that with the lowest and the resistance of $e$ is $\bmw(e)$, and in the directed case, each hyperarc $e=(T_e,S_e) \in E$ acts as an imaginary unit such that the current flows from the vertex in $T_e$ with the highest potential to the vertex in $H_e$ with the lowest and the resistance of $e$ is $\bmw(e)$.
The effective resistance is equal to the potential difference of $s$ and $t$ when an unit electricity is injected to $s$ and out from $t$.
The effective resistance of hypergraphs have applications in network analysis and semi-supervised learning (see~\cite{Fujii2018:submodLxb}).

An equivalent formulation of the effective resistance is the following.

\begin{lemma}
  Let $G=(V,E,\bmw)$ be an undirected/directed hypergraph.
  Then, we have
  \begin{align}\label{eq:effres-2}
    R_G(s,t) = \min \{ \bmx^\top L_G(\bmx) : \bmx(s) = 1, \bmx(t) = -1 \}.
  \end{align}
\end{lemma}
\begin{proof}
    Introducing Lagrange multipliers $\lambda_s$ and $\lambda_t$, an optimal solution $\bmx \in \bbR^V$ of~\eqref{eq:effres-2} must satisfy $2L_G(\bmx) - \lambda_s \bme_s - \lambda_t \bme_t = 0$.
    Therefore, $L_G(\bmx) = \frac{1}{2}(\lambda_s \bme_s + \lambda_t \bme_t)$.
    This equality has a solution if $(\lambda_s \bme_s + \lambda_t \bme_t)^\top \bmone =  0$, which implies $\lambda_s + \lambda_t = 0$.
    Then, we can rewrite the condition as $L_G(\bmx) = \frac{\lambda}{2}(\bme_s - \bme_t)$ for some $\lambda \in \bbR$.
    Therefore, $\bmx$ is a potential corresponding to an $st$-flow.
    Furthermore, since $\bmx(s) = 1$ and $\bmx(t) = -1$, the flow is a unit flow and we must have $\lambda = 2$.
    Hence an optimal solution $\bmx$ satisfies $L_G(\bmx) = \bme_s - \bme_t$, which is in turn the optimality condition of \eqref{eq:effres-1}.
    Evidently, the optimal value of \eqref{eq:effres-1} is equal to $\bmx^\top L_G(\bmx)$.
\end{proof}


An immediate corollary is the following.
\begin{corollary}
  If $H$ is an $\epsilon$-spectral sparsifier of $G$, then for distinct $s, t \in V$, we have $(1-\epsilon) R_H(s,t) \leq R_G(s,t) \leq (1+\epsilon) R_H(s,t)$.
\end{corollary}

\subsection{Faster Semi-Supervised Learning on Hypergraphs}
Zhang~\emph{et~al.}~\cite{Zhang:2017va} studied the following semi-supervised learning problem on directed hypergraphs.
Suppose that we are given a directed hypergraph $G=(V,E,\bmw)$, which represents a directional dependence of vertices.
We are given labels $\bmx^*(u)$ of vertices $u$ in $L \subseteq V$, and the task is to predict labels $\bmx(v) \in [0,1]$ in $ v \in V \setminus L$, respecting the structure of the hypergraph.
Intuitively, if a vertex $s$ is downstream\footnote{We say that $s$ is downstream of $t$ if there exists a directed path in $G$ from $t$ to $s$.} of another vertex $t$ in $G$, it is likely to hold that $\bmx(s) \leq \bmx(t)$.
The formulation considered in~\cite{Zhang:2017va} is as follows:
\begin{align*}
    \min        & \quad \sum_{e=(T_e,H_e)\in E} \bmw(e)\max_{u \in T_e}\max_{v \in H_e}\bigl([\bmx(u) - \bmx(v)]^+\bigr)^2 \\
    \text{s.t.} & \quad \bmx(u) = \bmx^*(u) \quad (u \in L).
\end{align*}
Since the objective function is $\bmx^\top L_G(\bmx)$, we can use our spectral sparsifiers for solving this problem faster.
More specifically, in~\cite{Zhang:2017va}, they proposed a subgradient descent algorithm for this problem, which computes a subgradient at a given point in $O(\size(G))$ time.
For simplicity, assume that $G$ is $r$-regular.
Then our hypergraph sparsification reduces the complexity of computing a subgradient to $\widetilde{O}(n^3 r)$ time, which is an improvement if $m = \omega(n^3)$.

\bibliographystyle{abbrv}
\bibliography{main}

\appendix

\section{A Hypergraph That Effective Resistance Fails to Sparsify}
\label{sec:resistance}
We provide an example for which a natural sparsification strategy based on effective resistance fails.
More precisely, let us consider a sampling algorithm such that each hyperedge $e$ is sampled with probability $p_e := \max_{\pi} R_{G_\pi}(e_\pi)$, where $R_{G_\pi}(e_\pi)$ denotes the effective resistance between the endpoints of $e_\pi$ in the graph $G_\pi$ (see Section~\ref{sec:hypergraph} for the definitions of $G_\pi$ and $e_\pi$).
In the following, we show that $p_e \geq 1$ for any hyperedge $e$ in the worst case, which means that the algorithm does not sparsify it at all.

Let $r$ and $n$ be positive integers.
Let $V = \{s, t\} \cup U$, where $U$ is a finite set of size $n$.
Let $E = \{ \{s, t\} \cup X : X \subseteq U, \abs{X} = r \}$.
Fix an arbitrary hyperedge $e = \{s, t\} \cup X$.
Define $\bmx \in \bbR^V$ as follows:
\begin{align*}
        \bmx(v) :=
    \begin{cases}{}
        1 & v = s \\
        1 - \epsilon & v \in X \\
        1 + \epsilon & v \notin X \\
        0 & v = t,
    \end{cases}
\end{align*}
where $\epsilon \in (0,1)$ is a constant.
Let $\pi$ be an ordering of $V$ in the descending order of $\bmx$ (break ties arbitrary).
Then, $G_\pi$ has only one edge connecting $s$ and $t$, which is realized by $e$, and the other edges do not connect $s$ and $t$ in $G_\pi$.
Therefore, $R_{G_\pi}(e_\pi) = 1$ for this ordering $\pi$.
Since $e$ was taken to be arbitrary, we must have $\max_{\pi} R_{G_\pi}(e_\pi) \geq 1$ for any hyperedge $e$.

\end{document}